\documentclass[a4paper]{birkart}
\date{February 5, 2007}

\usepackage{amsxtra}
\usepackage[dvips]{graphicx}

\newtheorem{theorem}{Theorem}[section]
\newtheorem{proposition}[theorem]{Proposition}
\newtheorem{lemma}[theorem]{Lemma}
\newtheorem{corollary}[theorem]{Corollary}

\theoremstyle{definition}

\theoremstyle{remark}

\newtheorem{remark}[theorem]{Remark}

\numberwithin{equation}{section}

\newcommand{\A}{\mathbf{A}}

\newcommand{\C}{\mathbb{C}}
\newcommand{\D}{\mathbf{D}}
\renewcommand{\epsilon}{\varepsilon}

\newcommand{\N}{\mathbb{N}}

\renewcommand{\phi}{\varphi}

\newcommand{\R}{\mathbb{R}}

\newcommand{\Z}{\mathbb{Z}}

\DeclareMathOperator{\erf}{erf}

\DeclareMathOperator{\Si}{Si}

\DeclareMathOperator{\tr}{tr}

\title[Eigenvalue estimates for the Aharonov-Bohm operator in a domain]{Eigenvalue estimates \\ for the Aharonov-Bohm operator in a domain}

\author{Rupert L. Frank}
\author{Anders M. Hansson}
\address{Department of Mathematics, School of Engineering Sciences,
  Royal Institute of Technology,   100 44 Stockholm, Sweden}
\email{\{rupert, anhan\}@math.kth.se}

\begin{document}

\begin{abstract}
  We prove semi-classical  estimates on moments of eigenvalues of the
  Aharonov-Bohm 
  operator in bounded two-dimensional domains. Moreover, we present a
  counterexample to the generalized diamagnetic inequality which was proposed
  by Erd\H os, Loss and Vougalter. Numerical studies complement these
  results. 
\end{abstract}

\maketitle


\section{Introduction}

We shall study inequalities for the eigenvalues of the Aharonov-Bohm operator
\begin{equation}\label{eq:abop}
	H_\alpha^\Omega := (\D-\alpha\A_0)^2
	\qquad\mbox{in}\ L_2(\Omega).
\end{equation}
Here $\Omega\subset\R^2$ is a bounded domain, $\D:=-i\nabla$ and
$\alpha\A_0(x):=\alpha |x|^{-2}(-x_2,x_1)^T$ is a vector potential
generating an Aharonov-Bohm magnetic field with flux $\alpha$ through the
origin. We shall assume that this point belongs to the interior of the simply-connected hull of $\Omega$ and that $\alpha\not\in\Z$, for otherwise $\alpha\A_0$ can be gauged away. On the boundary of $\Omega$ we impose Dirichlet boundary conditions. More precisely, the operator \eqref{eq:abop} is defined through the closure of the quadratic form $\|(\D-\alpha\A_0)u\|^2$ on $C_0^\infty(\Omega\setminus\{0\})$.

Before stating our main results we would like to recall some
well-known semi-classical spectral asymptotics and estimates for the
Dirichlet Laplacian $-\Delta^\Omega$ and its magnetic version
$(\D-\A)_\Omega^2$, $\A$ being an arbitrary vector potential. If $\Omega$ is
bounded then the spectrum of $-\Delta^\Omega$ is discrete, and by a
classical result due to Weyl (see, e.g., \cite{ReSi}) one has, as
$\Lambda\to\infty$, 
\begin{equation}\label{eq:weyl}
	\tr(-\Delta^\Omega-\Lambda)_-^\gamma 
	\sim \frac1{(2\pi)^2} \iint_{\Omega\times\R^2} (|\xi|^2-\Lambda)_-^\gamma\,dx\,d\xi
	= \frac1{4\pi(\gamma+1)}|\Omega|\Lambda^{\gamma+1}
\end{equation}
for all $\gamma\geq 0$. Note that the right-hand side involves the symbol $|\xi|^2$ on the phase space $\Omega\times\R^2$. The asymptotics \eqref{eq:weyl} are accompanied by the estimate
\begin{equation}\label{eq:bly}
	\tr(-\Delta^\Omega-\Lambda)_-^\gamma 
	\leq R_\gamma \frac1{(2\pi)^2} \iint_{\Omega\times\R^2} (|\xi|^2-\Lambda)_-^\gamma\,dx\,d\xi,
	\qquad \gamma\geq 0,
\end{equation}
with a universal constant $R_\gamma$ independent of $\Omega$ and
$\Lambda$, and one is interested in the sharp value of this constant
$R_\gamma$. In view of \eqref{eq:weyl} the sharp constant obviously
cannot be smaller than 1, and by an argument of Aizenman and Lieb
\cite{AL}, it is a non-increasing function of $\gamma$. P\'olya
\cite{P} proved the estimate \eqref{eq:bly} for $\gamma=0$ with
constant $1$ \emph{under the additional assumption that $\Omega$ is a
  tiling domain}. His famous conjecture that this is true for
arbitrary domains is still unproved. Berezin \cite{B} and
independently Li and Yau \cite{LY} (see also \cite{laptev}) proved
\eqref{eq:bly} for $\gamma\geq 1$ with the sharp constant
$R_\gamma=1$. This also yields the so far best known bound on the
sharp constant  for $\gamma=0$, namely $R_0\leq 2$. Indeed,
\begin{equation}\label{eq:goingdown}
	\tr(-\Delta^\Omega-\Lambda)_-^0 
	\leq (\mu-\Lambda)^{-1} \tr(-\Delta^\Omega-\mu)_-
	\leq (\mu-\Lambda)^{-1} \frac1{8\pi}|\Omega|\mu^{2},
	\qquad \mu>\Lambda,
\end{equation}
and the claim follows by optimization with respect to $\mu$. We note
that the estimate \eqref{eq:bly} for $\gamma=0$ and $\gamma=1$ is
closely related to the estimates 
\begin{equation*}
	\lambda_N^\Omega \geq \rho_0 4\pi |\Omega|^{-1}N
	\qquad\text{and}\qquad
	\sum_{j=1}^N \lambda_j^\Omega \geq \rho_1 2\pi |\Omega|^{-1} N^2
\end{equation*}
for the eigenvalues $\lambda_j^\Omega$ of the operator
$-\Delta^\Omega$. The form \eqref{eq:bly}, however, shows the close
connection with the Lieb-Thirring inequality, see \cite{LT1} and also
the review article \cite{LW2}. 

We  now turn to the `magnetic' analog of \eqref{eq:bly}, i.e., where
$-\Delta^\Omega$ is replaced by the Dirichlet realization of the
operator $(\D-\A)^2_\Omega$ in $L_2(\Omega)$ and $\A$ is a
(sufficiently regular) magnetic vector potential. Note that the value
of the right-hand side in \eqref{eq:bly} does not change if $\xi$ is replaced by $\xi-\A(x)$. Hence one is interested in the estimate
\begin{equation}\label{eq:blymag}
	\tr((\D-\A)^2_\Omega-\Lambda)_-^\gamma 
	\leq R_\gamma^{\rm{mag}} \frac1{(2\pi)^2} \iint_{\Omega\times\R^2} (|\xi|^2-\Lambda)_-^\gamma\,dx\,d\xi,
	\qquad \gamma\geq 0,
\end{equation}
with a universal constant $R_\gamma^{\rm{mag}}$ independent of
$\Omega$, $\Lambda$ and $\A$. It is a consequence of the sharp
Lieb-Thirring inequality by Laptev and Weidl \cite{LW} that
$R_\gamma^{\rm{mag}}=1$ for $\gamma\geq \frac32$. Not much is known
about \eqref{eq:blymag} in the case $\gamma<\frac32$. The Laptev-Weidl
result and an argument similar to \eqref{eq:goingdown} yield the
(probably non-sharp) estimate $R_{\gamma}^{\rm{mag}} \leq (\frac53)^{3/2}
(\gamma/(\gamma+1))^\gamma$ for $\gamma<\frac32$. For $\gamma=0$ and
$\gamma=1$ in particular one finds the values $2.1517$ and $1.0758$,
respectively. In \cite{ELV} the estimate \eqref{eq:blymag} is shown to
hold for $\gamma\geq1$ with constant $1$ \emph{in the special case of
  a homogeneous magnetic field} 
\begin{equation}\label{eq:ahom}
	\A(x) = \frac B2 (-x_2,x_1)^T.
\end{equation} 
It was recently shown in \cite{FLW} that \eqref{eq:blymag} does
\emph{not} hold with constant $1$ if $0\leq\gamma<1$, not even when
$\Omega$ is tiling. 
Moreover, the authors determined the optimal constant such that
\eqref{eq:blymag}  holds
for all $\Lambda$ under the constraint that $\A$ is given by
\eqref{eq:ahom} and $\Omega$ is tiling. 


In this paper we shall consider $\A$ corresponding to an Aharonov-Bohm magnetic field and we shall prove the estimate
\begin{equation}\label{eq:blyab}
	\tr(H_\alpha^\Omega-\Lambda)_-^\gamma 
	\leq C_\gamma(\alpha) \frac1{(2\pi)^2} \iint_{\Omega\times\R^2} (|\xi|^2-\Lambda)_-^\gamma\,dx\,d\xi,
	\qquad \gamma\geq 1,
\end{equation}
with a constant $C_\gamma(\alpha)$ given explicitly in terms of Bessel
functions. Even though our bound is probably not sharp, it improves
upon the previously known estimates. Indeed, numerical evaluation of
our constant shows that \eqref{eq:blyab} holds for all $\alpha$ with
constants $C_0(\alpha)=1.0540$ and $C_1(\alpha)=1.0224$ if $\gamma=0$
and $1$, respectively, see Section \ref{sec:bly}. We complement our
analytical results with a numerical study of the eigenvalue of the
operator \eqref{eq:abop} for five domains: a disc, a square and three
different annuli. In all cases the estimate \eqref{eq:blyab} seems to
be valid with constant $1$. We refer to Section \ref{sec:num} and
Figures \ref{fig:ltbot}--\ref{fig:rf3a} for a detailed account of
the outcome of our experiments. 

For the proof of our eigenvalue estimate we proceed similarly as in \cite{ELV}. Indeed, by the Berezin-Lieb inequality, \eqref{eq:blyab} is an immediate consequence of the \emph{generalized diamagnetic inequality}
\begin{equation}\label{eq:gendiamag}
\tr\chi_\Omega(H_\alpha-\Lambda)_-^\gamma\chi_\Omega 
\leq R_{\gamma}(\alpha) \tr\chi_\Omega(-\Delta-\Lambda)_-^\gamma\chi_\Omega
\qquad
\mbox{for all bounded $\Omega\subset\R^2$}.
\end{equation}
Here $H_\alpha:=H_\alpha^{\R^2}$ denotes the Aharonov-Bohm operator in
the whole space. In \cite{ELV} an analogous estimate was proved in the
case \eqref{eq:ahom} \emph{with constant} 1 when $\gamma\ge 1$. 
The authors conjectured that such an inequality is not true for an
arbitrary magnetic field, but their counterexample contains a gap;
the condition 2) on p. 905 can not be satisfied by a non-trivial
radial vector field, as it was first pointed out by M.~Solomyak.
This gap can be removed by a minor change in the argument,
since the assumption of radial symmetry is not essential in the proof 
\cite{ELo}. Nevertheless, we
feel that our example is of independent interest and sheds some light
on the particularities of the Aharonov-Bohm operator. What we prove is
that the sharp constant in \eqref{eq:gendiamag} is strictly greater
than unity, see Theorem \ref{diamagthm}. We establish this by a thorough
study of the local spectral density of the operator $H_\alpha$, see
Section~\ref{sec:density}. 

We mention in closing the papers \cite{bel}, \cite{ef}, \cite{amh}, \cite{mor},
where Lieb-Thirring estimates for the Schr\"odinger operator
$H_\alpha+V$ were obtained. Our estimates can be seen as a refinement
of these estimates in the special case where the potential $V$ is
equal to a negative constant $-\Lambda$ inside and equal to infinity
outside a bounded domain $\Omega\subset\R^2$. 

\textbf{Acknowledgements.} The authors would like to thank A. Laptev for the setting of the problem and helpful remarks. The first author is grateful to E. H. Lieb and R. Seiringer for their hospitality at Princeton University and thanks them, H. Kalf and M. Loss for fruitful discussions.


\section{The Aharonov-Bohm operator in the whole space}\label{sec:density}

\subsection{Diagonalization}

In this section we recall some well-known facts about the
Aharonov-Bohm operator in the whole space, see, e.g., \cite{ab}, \cite{ruij}. We
denote by $H_\alpha$ the self-adjoint operator in $L_2(\R^2)$
associated with the closure of the quadratic form 
\begin{equation*}
	\int_{\R^2} |(\D-\alpha\A_0)u|^2\,dx,
	\qquad u\in C_0^\infty(\R^2\setminus\{0\}).
\end{equation*}
Here $\D:=-i\nabla$, $\A_0(x):=|x|^{-2}(-x_2,x_1)^T$ and
$\alpha\in\R$. Moreover,  $J_\nu$ denotes, as usual, the Bessel function of the first kind of order $\nu$, see \cite{as}. With polar coordinates $x=|x|(\cos\theta_x,\sin\theta_x)$ and similarly for $\xi$, we define
\begin{equation*}
	\mathcal F_\alpha(\xi,x) := \frac1{2\pi}\sum_{n\in\Z} J_{|n-\alpha|}(|\xi||x|) e^{in(\theta_\xi-\theta_x)},
	\qquad \xi,x\in\R^2,
\end{equation*}
and put 
\begin{equation*}
	(\mathcal F_\alpha u)(\xi) := \int_{\R^d} \mathcal F_\alpha(\xi,x) u(x)\,dx,
	\qquad \xi\in\R^2,
\end{equation*}
for $u\in C_0^\infty(\R^2)$. Note that $\mathcal
F_0(\xi,x)=(2\pi)^{-1} e^{-i\xi\cdot x}$ \cite[9.1.41]{as}, so
$\mathcal F_0$ is the ordinary Fourier transform, which diagonalizes
$H_0=-\Delta$. Similarly, one has 

\begin{lemma}\label{diagonal}
	For any $\alpha\in\R$, $\mathcal F_\alpha$ extends to a unitary operator in $L_2(\R^2)$ and diagonalizes $H_\alpha$, i.e.,
	\begin{equation*}
		(\mathcal F_\alpha f(H_\alpha)u)(\xi)= f(|\xi|^2) (\mathcal F_\alpha u)(\xi),
		\qquad \xi\in\R^2,
	\end{equation*}
for any $u\in L_2(\R^2)$ and $f\in L_\infty(\R)$.
\end{lemma}
	
We sketch a proof of this assertion for the sake of completeness.

\begin{proof}
	The orthogonal decomposition
	\begin{equation*}
		L_2(\R^2) = \bigoplus_{n\in\Z} \mathfrak H_n,
		\qquad \mathfrak H_n:=\{ |x|^{-1/2}g(|x|) e^{in\theta_x} : \ g\in L_2(\R_+) \},
	\end{equation*} 
	reduces $H_\alpha$. The part of $H_\alpha$ in $\mathfrak H_n$ is unitarily equivalent to the operator
	\begin{equation*}
		h_{|n-\alpha|} := -\frac{d^2}{dr^2}+\frac{(n-\alpha)^2-1/4}{r^2}
		\qquad\mbox{in } L_2(\R_+),
	\end{equation*}
which is	defined as the Friedrichs extension of the corresponding differential expression on $C^\infty_0(\R_+)$. (We emphasize that in our notation $\R_+$ means the \emph{open} interval $(0,\infty)$.) The operator
	\begin{equation*}
		(\Phi_\nu g)(k) := \int_{\R_+} \sqrt{rk} J_\nu(rk) g(r) \,dr,
		\qquad k\in\R_+,
	\end{equation*}
	initially defined for $g\in C^\infty_0(\R_+)$, extends to a unitary operator in $L_2(\R_+)$ and diagonalizes $h_\nu$, i.e.,
	\begin{equation*}
		(\Phi_\nu f(h_\nu)g)(k)= f(k^2) (\Phi_\nu g)(k),
		\qquad k\in\R_+,
	\end{equation*}
for any $g\in L_2(\R_+)$ and $f\in L_\infty(\R)$
\cite[Ch.~VIII]{T}. The assertion of the lemma is a simple consequence
of these facts. 
\end{proof}

The proof of the preceding lemma shows in particular that the
operators $H_\alpha$ and $H_{\alpha+m}$ with $m\in\Z$ are unitarily
equivalent via multiplication by $e^{im\theta_x}$ (a gauge
transformation). Hence, without loss of generality \emph{we shall
   assume that $0\leq\alpha<1$.}  

Lemma \ref{diagonal} implies that $f(H_\alpha)$, at least formally, is
an integral operator with integral kernel 
\begin{align*}
	f(H_\alpha)(x,y)
	& = \int_{\R^2} \overline{\mathcal F_\alpha(\xi,x)} f(|\xi|^2) \mathcal F_\alpha(\xi,y)\,d\xi \\
	& = \frac1{2\pi}\sum_{n\in\Z} \int_0^\infty f(k^2) 
	J_{|n-\alpha|}(k|x|) J_{|n-\alpha|}(k|y|) e^{in(\theta_x-\theta_y)}k \, dk \\
	& = \frac1{4\pi}\sum_{n\in\Z} \int_0^\infty f(\lambda) 
	J_{|n-\alpha|}(\sqrt\lambda|x|) J_{|n-\alpha|}(\sqrt\lambda|y|) e^{in(\theta_x-\theta_y)} \, d\lambda.
\end{align*}
On the diagonal this is
\begin{equation}\label{eq:diagonal}
	f(H_\alpha)(x,x)
	= \frac1{4\pi} \int_0^\infty f(\lambda) \rho_\alpha (\sqrt\lambda|x|) \, d\lambda,
\end{equation}
where
\begin{equation}\label{eq:rho}
	\rho_\alpha(t) := \sum_{n\in\Z} J_{|n-\alpha|}^2(t),
	\qquad t\geq 0.
\end{equation}
In particular, $\frac1{4\pi} \rho_\alpha (\sqrt\lambda|x|)$ is the
\emph{local spectral density} at energy $\lambda$. In the following
subsection we collect some basic information about this function, and
in Subsections \ref{sec:mom} and \ref{sec:exp} we prove some results
about the precise asymptotic behavior of \eqref{eq:diagonal} as
$|x|\to\infty$ for special choices of functions $f$. This will allow
us to prove that the generalized diamagnetic inequality is violated.


\subsection{The spectral density}

Our results are based on a detailed study of the function
$\rho_\alpha$ from \eqref{eq:rho}. We note that if $\alpha=0$ then
$\rho_0\equiv 1$ by \cite[9.1.76]{as}. An expression in terms of
elementary functions is also available \cite[5.2.15]{as} for $\alpha=1/2$,
\begin{equation}\label{eq:rhohalf}
	\rho_{1/2}(t)=\frac2\pi\int_0^{2t}\frac{\sin s}s\,ds.
\end{equation}
As $t\to\infty$, $\rho_{1/2}(t)$ tends to $1$ in an oscillating manner. As we shall see, this behavior appears for all $0<\alpha<1$. The starting point of our study of $\rho_\alpha$ with non-trivial flux $\alpha$ is the following

\begin{lemma}\label{rhoderivative}
	For any $0<\alpha<1$, $\rho_\alpha$ is a smooth function on $\R_+$ with $\rho_\alpha(0)=0$, $\rho_\alpha(t)\to 1$ as $t\to\infty$ and
	\begin{equation}\label{eq:rhoderivative}
		\rho_\alpha'(t) 
		= J_\alpha(t)J_{\alpha-1}(t)+J_{1-\alpha}(t)J_{-\alpha}(t),
		\qquad t\geq 0.
	\end{equation}        
\end{lemma}

\begin{proof}
	By \cite[11.2(10)]{luke} and \cite[11.4.42]{as} one has, for
        all $t\geq 0$, 
	\begin{align*}
		\rho_\alpha(t) 
		& = \int_0^t \left(J_\alpha(s)J_{\alpha-1}(s)+J_{1-\alpha}(s)J_{-\alpha}(s)\right)\,ds \\
		& = 1- \int_t^\infty \left(J_\alpha(s)J_{\alpha-1}(s)+J_{1-\alpha}(s)J_{-\alpha}(s)\right)\,ds,
	\end{align*}
	which implies the assertion.
\end{proof}


Our next result will not be needed in the sequel, but it helps to
clarify the behavior of $\rho_\alpha$ and demonstrates the methods which we
shall use later on.

\begin{lemma}\label{rem22}
	Let $0<\alpha<1$. As $t\to\infty$, 
	\begin{equation*}
    \rho_\alpha(t)=1-\frac{\sin\alpha\pi}{\pi}\frac{\cos 2t}{t}+
    \mathcal O(t^{-2}).
  \end{equation*}
\end{lemma}

\begin{proof}
	The asymptotics \cite[9.2.5]{as}
	\begin{equation*}\label{eq:besselasy}
  	J_\nu(t)=\sqrt{\frac2{\pi t}} 
  	\left(\cos\left(t-\frac{\pi}4-\frac{\nu\pi}2\right)
  	-\frac{4\nu^2-1}{8t}\sin\left(t-\frac{\pi}4-\frac{\nu\pi}2\right)
  	+\mathcal O (t^{-2}) \right),
	\end{equation*}
	the formula \eqref{eq:rhoderivative} and elementary manipulations show that
	\begin{equation*}
		\rho_\alpha'(t)
		= \frac2{\pi t}
		\left(\sin\alpha\pi\sin2t + \frac{(2\alpha-1)^2\sin\alpha\pi}{4}\frac{\cos2t}{t} 
		+ \mathcal O (t^{-2}) \right).
	\end{equation*}
	Using that $\rho_\alpha(t) = 1-\int_t^\infty
        \rho_\alpha'(s)\,ds$ by Lemma \ref{rhoderivative},	we
        obtain the assertion by repeated integration by parts.
\end{proof}

\subsection{Moments of the spectral density}\label{sec:mom}

For any $\gamma>-1$ let us define 
\begin{equation}\label{eq:sigma}
  \sigma_{\alpha,\gamma}(r):=
  \int_0^1(1-\lambda)^\gamma\rho_\alpha(\sqrt\lambda r) d\lambda,
  \qquad r\geq 0.
\end{equation}
We are interested in the asymptotic behavior of this quantity or, more precisely, in the way it approaches its limit.

\begin{theorem}\label{sigmathm}
	Let $0<\alpha<1$ and $\gamma>-1$. As $r\to\infty$,
  \begin{equation*}
    \sigma_{\alpha,\gamma}(r)=
		\frac1{\gamma+1} 
		-\Gamma(\gamma+1)\frac{\sin\alpha\pi}{\pi}
		\frac{\sin(2r-\frac12{\gamma\pi})}{r^{2+\gamma}}+
                \mathcal O(r^{-3-\gamma}).
  \end{equation*}
\end{theorem}

For the proof we note that by Lemma \ref{rhoderivative} and dominated convergence
\begin{equation*}
	\lim_{r\to\infty} \sigma_{\alpha,\gamma}(r) 
	= \int_0^1 (1-\lambda)^\gamma \, d\lambda = \frac1{\gamma+1}.
\end{equation*}
In view of
\begin{equation*}\label{eq:sigmaproof12}
	\sigma_{\alpha,\gamma}(r) = \frac1{\gamma+1} - \int_r^\infty \sigma_{\alpha,\gamma}'(s)\, ds,
\end{equation*}
Theorem \ref{sigmathm} follows via integration by parts from

\begin{proposition}\label{sigmaderivative}
	Let $0<\alpha<1$ and $\gamma>-1$. As $r\to\infty$,
  \begin{equation}\label{eq:sigmaderivative}
    \sigma_{\alpha,\gamma}'(r)=
		\frac{\Gamma(\gamma+1)}{r^{\gamma+2}}\frac{\sin \alpha\pi}{\pi}
		\left( -2\cos\left(2r-\frac{\gamma\pi}2\right) 
		-\frac{d_1}{r}\sin\left(2r-\frac{\gamma\pi}2\right)
		+\mathcal O(r^{-2})\right)
  \end{equation}
  where $d_1:=\frac18(1-4((\gamma+\frac32)^2+(1-2\alpha)^2)$.
\end{proposition}

We defer the rather technical proof of this proposition to the
following subsection. We would like to point out that
\eqref{eq:sigmaderivative} can be proved in an elementary way if
$\alpha=\frac12$ due to the formula \eqref{eq:rhohalf}. Indeed, if we
simplify further by taking  $\gamma=1$, then 
\begin{align*}
	\sigma_{1/2,1}'(r) 
	& = \frac2{\pi r} \int_0^1 (1-\lambda) \sin (2\sqrt\lambda r) \, d\lambda 
 	= \frac4{\pi r} \int_0^1 k(1-k^2) \sin 2kr \, dk \\
	& = -\frac2{\pi r} \left(\frac{\sin 2r}{r^2} + \frac{3\cos 2r}{2r^3}
  + \mathcal O({r^{-4}})\right) 
\end{align*}
by integration by parts. This is \eqref{eq:sigmaderivative} in this
special case.


\subsection{Proof of Proposition \ref{sigmaderivative}}

We shall need

\begin{lemma} \label{inteq}
  Let $\nu>-1/2$ and $\gamma>-1$. Then
  \begin{equation*}
    \int_0^1
    J_{\nu+1/2}(kr)J_{\nu-1/2}(kr)k^2(1-k^2)^\gamma \,dk  
    =c_{\gamma}\int_0^{1}
    J_{2\nu}(2kr)k (1-k^2)^{\gamma+1/2}dk
  \end{equation*}
  with
  \begin{equation}\label{eq:cgamma}
    c_{\gamma}
    :=\frac1{\sqrt\pi}\frac{\Gamma(\gamma+1)}{\Gamma(\gamma+\frac32)}.
  \end{equation}
\end{lemma}

\begin{proof}
  By virtue of the series representations \cite[9.1.10, 9.1.14]{as}
  for $\alpha>0$,
  \begin{align*}
    J_{\alpha-1}(t)&=
    \left(\frac{t}2\right)^{\alpha-1}\sum_{n=0}^\infty
    \frac{(-1)^n}{n!\Gamma(\alpha+n)}
    \left(\frac{t}2\right)^{2n},  \\
    J_{\alpha}(t)J_{\alpha-1}(t)&=
    \left( \frac{t}2\right)^{2\alpha-1} \sum_{n=0}^\infty
    \frac{(-1)^n\Gamma(2\alpha+2n)}{n!\Gamma(\alpha+n)\Gamma(\alpha+n+1)
      \Gamma(2\alpha+n)} \left(\frac{t}2\right)^{2n},
  \end{align*}
  the statement is equivalent to
  \begin{multline}\label{eq11}
    \left(\frac{r}2\right)^{2\nu-1}\sum_{n=0}^\infty
    \frac{(-1)^n(r/2)^{2n}\Gamma(2\nu+2n)}
    {n!\Gamma(2\nu+n)\Gamma(\nu+n)\Gamma(\nu+n+1)}
    \int_0^1 k^{2\nu+2n+1}(1-k^2)^\gamma dk
    \\ =
    \frac{r^{2\nu-1}}{\sqrt\pi}
    \frac{\Gamma(\gamma+1)}{\Gamma(\gamma+\frac32)}
    \sum_{n=0}^\infty
    \frac{(-1)^nr^{2n}}{n!\Gamma(2\nu+n)}
    \int_0^1k^{2\nu+2n}(1-k^2)^{\gamma+1/2} dk,
  \end{multline}
an equality which actually holds termwise. This follows by the beta function
  identity
  \begin{equation*}
    2 \int_0^1 k^{2\alpha-1}(1-k^2)^{\beta-1} dk =
    \int_0^1 t^{\alpha-1}(1-t)^{\beta-1}dt =B(\alpha,\beta)
    =\frac{\Gamma(\alpha)\Gamma(\beta)}{\Gamma(\alpha+\beta)}
  \end{equation*}
  and the duplication formula
	$\sqrt\pi\Gamma(2\alpha)=2^{2\alpha-1} \Gamma(\alpha)\Gamma\left(\alpha+\frac12\right)$
	\cite[6.1.18]{as}. 
\end{proof}

\begin{lemma}\label{decomplemma}
  Let $\beta>0$ and $-2<\nu<2$. The following decomposition holds:
  \begin{equation}\label{eq10}
    \int_0^1J_\nu(kr)k(1-k^2)^{\beta-1} dk
    =f_{\beta,\nu}(r)+g_{\beta,\nu}(r),
  \end{equation}
  where $g_{\beta,-\nu}(r)=-g_{\beta,\nu}(r)$ and
  \begin{equation*}
    f_{\beta,\nu}(r)=
    \frac{\Gamma(\beta)}{2\sqrt\pi}\left(\frac2r\right)^{\beta+1/2}\left[
      \cos(r-r_0)+\frac{d_1}r\sin(r-r_0)+\mathcal O(r^{-2})\right],
    \qquad r\to\infty,
  \end{equation*}
  with 
  \begin{equation}\label{eq:r0d1}
    r_0:=\frac\pi 2 \left(\beta+\nu+\frac12\right)
    \quad\text{and}\quad
    d_1:=\frac{1}8-\frac{\beta^2+\nu^2}2.
  \end{equation}
\end{lemma}

The function $g_{\beta,\nu}$ has a power-like, non-oscillatory
  behavior at infinity and dominates $f_{\beta,\nu}$ when
  $\beta>\frac32$. Its   odd parity with respect to $\nu$  is,
  however, the property vital to   us, since it leads to a useful
  cancellation. 

\begin{proof}
  The left-hand side in \eqref{eq10} can be expressed  as a
  generalized hypergeometric function. Recall that for $p,q\in\N_0$,
  $p\leq q$ and
  $\alpha_1,\ldots,\alpha_p,\beta_1,\ldots,\beta_q\in\R\setminus(-\N_0)$,
  \begin{equation*}
    {}_pF_q\left(\left.
    \begin{array}{c}
      \alpha_1,\alpha_2,\ldots,\alpha_p \\
      \beta_1,\beta_2,\ldots,\beta_q
    \end{array}\right|
     z\right)=\sum_{n=0}^\infty
    \frac{(\alpha_1)_n(\alpha_2)_n \cdots (\alpha_p)_n}
	 {(\beta_1)_n(\beta_2)_n\cdots (\beta_q)_n}
    \frac{z^n}{n!},
    \qquad z\in\C,
  \end{equation*}
  where $(\alpha)_n=\Gamma(\alpha+n)/\Gamma(\alpha)
  =\alpha(\alpha+1)\cdots(\alpha+n-1)$. By \cite[13.3.2(10)]{luke},
  \begin{align*}
&    \int_0^1J_\nu(kr)k(1-k^2)^{\beta-1} dk
    = \int_0^{\pi/2} J_\nu(r\sin\theta)\cos^{2\beta-1}\theta\sin\theta\,
    d\theta \\ =
&    \frac{r^\nu \Gamma(\beta)\Gamma(\nu/2+1)
    }{2^{\nu+1}\Gamma(\nu+1)\Gamma(\beta+\nu/2+1)} {}_1F_2\left(
    \left. \begin{array}{c}
      \nu/2+1\\ \nu+1, \beta+\nu/2+1
    \end{array}\right| 
     -\frac{r^2}4 \right).
  \end{align*}
  Next we use the asymptotics of the generalized hypergeometric
  function \cite[1.3.3(5, 7, 8, 13)]{luke},
  \begin{align*}
&    \frac{\Gamma(\nu/2+1)}{\Gamma(\nu+1)\Gamma(\beta+\nu/2+1)} {}_1F_2\left(
    \left. \begin{array}{c}
      \nu/2+1\\ \nu+1, \beta+\nu/2+1
    \end{array}\right| 
     -\frac{r^2}4 \right) \\ \sim &
    \frac{(2/r)^{\beta+\nu+1/2}}{2\sqrt\pi} \left[
      e^{i(r-r_0)}\sum_{n=0}^\infty d_n\left({ir}\right)^{-n}
      +e^{-i(r-r_0)}\sum_{n=0}^\infty d_n\left(-{ir}\right)^{-n}
      \right] \\
  &  + \frac{\nu(2/r)^{\nu+2}}{2\Gamma(\beta)}
    {}_3F_0\left(\left. 
    \begin{array}{c}
      1+\nu/2, 1-\nu/2,1-\beta \\
      -
    \end{array}
    \right|-\frac4{r^2} \right),
    \qquad  r\to\infty,
  \end{align*}
  with $r_0$, $d_1$ as in \eqref{eq:r0d1} and $d_0=1$. We emphasize
  that the $\sim$ sign means equality in the sense of asymptotic
  expansions, and that ${}_3F_0$ is not a well-defined function but
  denotes an asymptotic expansion. We define $\nu^{-1}r^2g_{\beta,\nu}(r)$ as
  a finite approximation to ${}_3F_0$. Namely, 
  let $K$ denote the smallest non-negative integer such that
  $2K\geq\beta-\frac32$ and put
  \begin{equation*}
    g_{\beta,\nu}(r)
    := \frac\nu{r^2}
    \sum_{n=0}^K (1+\nu/2)_n(1-\nu/2)_n (1-\beta)_n
    \frac{(-1)^n}{n!} \left(\frac2r\right)^{2n}.
  \end{equation*}
  This function is antisymmetric in $\nu$ and
  \begin{align*}
    f_{\beta,\nu}(r)
    := & \int_0^1J_\nu(kr)k(1-k^2)^{\beta-1} dk-g_{\beta,\nu}(r)\\
    = & 
    \frac{r^\nu\Gamma(\beta)}{2^{\nu+1}}
    \frac{(2/r)^{\beta+\nu+1/2}}{2\sqrt\pi} \left[
      e^{i(r-r_0)} \left(d_0+\frac{d_1}{ir}\right)
      +e^{-i(r-r_0)} \left(d_0-\frac{d_1}{ir}\right) \right. \\
    &  \left. \phantom{\frac1\qquad} + \mathcal O(r^{-2}) \right] 
    + \mathcal O\left(r^{-2(K+2)}\right) \\
    =&\frac{\Gamma(\beta)}{2\sqrt\pi}\left(\frac2{r}\right)^{\beta+1/2}
    \left[
      d_0\cos(r-r_0)
      +\frac{d_1}r\sin(r-r_0)
      +\mathcal O (r^{-2})\right]
  \end{align*}
as claimed.
\end{proof}

Now everything is in place for the

\begin{proof}[Proof of Proposition~\ref{sigmaderivative}]
  By Lemmas \ref{rhoderivative} and \ref{inteq},
  \begin{align*}
    \sigma'_{\alpha,\gamma}(r)
    & = 2\int_0^1k^2(1-k^2)^\gamma\rho_\alpha'(kr)dk \\ 
    & =
    2\int_0^1(J_\alpha(kr)J_{\alpha-1}(kr)
    +J_{1-\alpha}(kr)J_{-\alpha}(kr))
    k^2(1-k^2)^\gamma dk\\
    & =
    2c_\gamma\int_0^1
    (J_{2\alpha-1}(2kr)+J_{1-2\alpha}(2kr))k(1-k^2)^{\gamma+1/2} dk
  \end{align*}
  with $c_\gamma$ from \eqref{eq:cgamma}. Now we apply Lemma
  \ref{decomplemma} with $\beta=\gamma+\frac32$ and
  $\nu=\pm(2\alpha-1)$. Using the antisymmetry of $g_{\beta,\nu}$ we
  find that
  \begin{align*}
    \sigma'_{\alpha,\gamma}(r)
    = 2c_\gamma 
    \left( f_{\gamma+3/2,2\alpha-1}(r)+f_{\gamma+3/2,1-2\alpha}(r)\right),
  \end{align*}
  and the assertion follows after elementary manipulations from the
  asymptotics of $f_{\beta,\nu}$ given in Lemma \ref{decomplemma}.
\end{proof}


\subsection{The Laplace transform of the spectral density}\label{sec:exp}

In this subsection we are interested in the quantity
\begin{equation*}
  \sigma_{\alpha,\infty}(r):=\int_0^\infty
  e^{-\lambda}\rho_\alpha(\sqrt{\lambda}r) d\lambda,
  \qquad r\geq 0.
\end{equation*}
Note that this is essentially the Laplace transform of the function $\rho_\alpha(\sqrt{\cdot})$.

\begin{theorem}\label{expthm}
  For all $0<\alpha<1$ the function $\sigma_{\alpha,\infty}$ is strictly increasing from $0$ to $1$ on $[0,\infty)$. In particular, $\sigma_{\alpha,\infty}(r)< 1$ for all $r\ge 0$.
\end{theorem}

This theorem shows that the oscillations we observed for $\sigma_{\alpha,\gamma}$ are no longer present. 

\begin{proof}
	Since $\rho_\alpha=\rho_{1-\alpha}$ and hence $\sigma_{\alpha,\infty}=\sigma_{1-\alpha,\infty}$, it suffices to treat the case $0<\alpha\leq 1/2$. By the properties of $\rho_\alpha$ (see Lemma \ref{rhoderivative}) and dominated convergence we get $\sigma_{\alpha,\infty}(0)=0$ and $\sigma_{\alpha,\infty}(r)\to 1$ as $r\to\infty$. Again by Lemma \ref{rhoderivative},
\begin{align*}
  \sigma_{\alpha,\infty}'(r)
  & = 2 \int_0^\infty k^2 e^{-k^2}\rho_\alpha'(kr)dk \\
	& = 2 \int_0^\infty k^2 e^{-k^2}\left(J_\alpha(kr)J_{\alpha-1}(kr)+J_{1-\alpha}(kr)J_{-\alpha}(kr)\right) dk,
\end{align*}
and hence according to \cite[13.4.1(10)]{luke},
\begin{multline*}
     \sigma'_{\alpha,\infty}(r)
= r^{1-2\alpha}  \left[ 
\frac{2^{1-2\alpha}}{r^{2(1-2\alpha)}\Gamma(\alpha)}
{}_1F_1\left(
  \left. \begin{array}{c}
    \alpha+\frac12 \\ 2\alpha
  \end{array}\right|
 -r^2\right) \right. \\ 
   \left.
+\frac{1}{2^{1-2\alpha}\Gamma(1-\alpha)}
{}_1F_1\left(
  \left. \begin{array}{c}
    \frac32-\alpha\\2-2\alpha
  \end{array}\right|
 -r^2 \right) \right].
\end{multline*}
In the special case $\alpha=1/2$ we note that
\begin{equation*}
\sigma'_{1/2,\infty}(r)
= \frac2{\sqrt\pi} {}_1F_1\left(
  \left. \begin{array}{c}
    1 \\ 1
  \end{array}\right|
 -r^2\right)
= \frac2{\sqrt\pi} e^{-r^2}.
\end{equation*}
If $0<\alpha<\frac12$ we can apply the \emph{Kummer transformations}
\cite[13.1.27]{as} (note that  ${}_1F_1(a,b;z)=M(a,b,z)$ in \cite{as})
to get 
\begin{multline*}
     \sigma'_{\alpha,\infty}(r)
= r^{1-2\alpha}e^{-r^2}  \left[
\frac{2^{1-2\alpha} }{r^{2(1-2\alpha)}\Gamma(\alpha)}
{}_1F_1\left(
  \left. \begin{array}{c}
    \alpha-\frac12 \\ 2\alpha
  \end{array}\right|
 r^2\right) \right. \\ 
 \left.
+\frac1{2^{1-2\alpha}\Gamma(1-\alpha)}
{}_1F_1\left(
  \left. \begin{array}{c}
    \frac12-\alpha\\2-2\alpha
  \end{array}\right|
 r^2 \right) \right].
\end{multline*}
By elementary properties of the gamma function this can be rewritten as
\begin{equation*}
	\sigma'_{\alpha,\infty}(r)
	=  \frac{2\sin\alpha\pi}{\sqrt\pi} r^{1-2\alpha} e^{-r^2}
        U\left(\frac12-\alpha,2-2\alpha,r^2\right) 
\end{equation*}
where \cite[13.1.3]{as}
\begin{multline*}
	U\left(\frac12-\alpha,2-2\alpha,r^2\right)
	=  \frac{\pi}{\sin2\alpha\pi} \frac1{r^{2(1-2\alpha)}}
	\left[
	\frac1{\Gamma(\frac12-\alpha)\Gamma(2\alpha)}
	{}_1F_1\left(
  \left. \begin{array}{c}
    \alpha-\frac12 \\ 2\alpha
  \end{array}\right|
	 r^2\right)
	\right. \\ 
	 \left.	
	- \frac{r^{2(1-2\alpha)}}{\Gamma(\alpha-\frac12)\Gamma(2-2\alpha)}
	{}_1F_1\left(
  \left. \begin{array}{c}
    \frac12-\alpha \\ 2-2\alpha
  \end{array}\right|
	 r^2\right)
	\right].
\end{multline*}
$U$ is positive by the integral representation  \cite[13.2.5]{as}, and
this proves the theorem. 
\end{proof}

\section{Counterexample to the generalized diamagnetic inequality}

Following \cite{ELV} we consider the question: \emph{Which
  non-negative convex functions $\phi$ vanishing at infinity satisfy }
\begin{equation}\label{eq:diamag}
	\tr\chi_\Omega\phi(H_\alpha)\chi_\Omega \leq \tr\chi_\Omega\phi(-\Delta)\chi_\Omega
	\quad\mbox{\emph{for all bounded domains  $\Omega\subset\R^2$?}}
\end{equation}
By \eqref{eq:diagonal} the statement \eqref{eq:diamag} is
equivalent to the pointwise inequality 
\begin{equation}\label{eq:diamagpointwise}
	\int_0^\infty \phi(\lambda) \rho_\alpha (\sqrt\lambda r) \, d\lambda 
	\leq \int_0^\infty \phi(\lambda) \, d\lambda,
	\quad\mbox{for all $r\in[0,\infty)$.}
\end{equation}
Note that \eqref{eq:diamag} is true for the family of functions
$\phi(\lambda)=e^{-t\lambda}$, $t>0$; we shall prove this (even with
strict inequality) in Remark~\ref{nyanm} below. Alternatively,
it follows from the  `ordinary' diamagnetic inequality (see, e.g.,
\cite{hus}), 
\begin{equation*}
  |\exp( -tH_\alpha)u| \leq \exp(-t(-\Delta))|u| \ \text{a.e.},
\qquad u\in L_2(\R^2),
\end{equation*}
and \cite[Thm.~2.13]{simo}, since
\begin{equation*}
  \left\| \chi_\Omega\exp\left(- \frac t2  H_\alpha\right) \right\|_2^2 
=\tr \chi_\Omega \exp(-t H_\alpha) \chi_\Omega.
\end{equation*}
We point out that the validity of \eqref{eq:diamag} for
$\phi(\lambda)=e^{-t\lambda}$, $t>0$,  actually implies that it holds
for any function of the form
\begin{equation}\label{eq:laplace}
	\phi(\lambda)=\int_0^\infty e^{-t\lambda} w(t)\,dt,
	\qquad w\geq 0.
\end{equation}

In connection with a Berezin-Li-Yau-type inequality one is particularly interested in the functions 
\begin{equation}\label{eq:momentphi}
	\phi(\lambda)=(\lambda-\Lambda)_-^\gamma,
	\qquad \gamma\geq 1,\ \Lambda>0.
\end{equation}
Note that these functions cannot be expressed in the form \eqref{eq:laplace}.

\begin{theorem}\label{diamagthm}
	Let $0<\alpha<1$ and let $\phi$ be given by
        $\eqref{eq:momentphi}$ for some $\gamma\geq 1$,
        $\Lambda>0$. Then the generalized diamagnetic inequality
        \eqref{eq:diamag} is violated. More precisely, there exist
        constants $C_1, C_2>0$ (depending on $\alpha$ and $\gamma$,
        but not on $\Lambda$) such that for all $|x|\geq
        C_1\Lambda^{-1/2}$, 
	\begin{equation}\label{eq:diamagthm}
		\left| \phi(H_\alpha)(x,x) - \phi(-\Delta)(x,x) 
		+ A_{\alpha,\gamma}(\Lambda)
	\frac{\sin(2\sqrt\Lambda|x|-\frac12\gamma \pi)}{|x|^{2+\gamma}} \right| 
	\leq C_2 \frac{\Lambda^{(\gamma-1)/2}}{|x|^{3+\gamma}}
  \end{equation}
  with $A_{\alpha,\gamma}(\Lambda):=(2\pi)^{-2} \Lambda^{\gamma/2}\Gamma(\gamma+1)\sin\alpha\pi$.
\end{theorem}

\begin{proof}
  By \eqref{eq:diagonal} and the scaling $\lambda\mapsto\Lambda\lambda$,
	\begin{equation}\label{eq:diamagthm1}
		\phi(H_\alpha)(x,x) 
		= \frac{\Lambda^{\gamma+1}}{4\pi} 
		\int_0^1 (1-\lambda)^\gamma \rho_\alpha (\sqrt{\Lambda\lambda}|x|) \, d\lambda
		= \frac{\Lambda^{\gamma+1}}{4\pi}\sigma_{\alpha,\gamma}(\sqrt\Lambda|x|),
	\end{equation}
	where we used the notation \eqref{eq:sigma}. Note that
        $\phi(-\Delta)(x,x)=(4\pi(\gamma+1))^{-1}\Lambda^{\gamma+1}$. The
        expansion \eqref{eq:diamagthm} is thus a consequence of
        Theorem \ref{sigmathm}. To prove that the generalized
        diamagnetic inequality \eqref{eq:diamag} is violated, one
        can consider the domains $\Omega_n := \{x\in\R^2 : \
        |\sqrt\Lambda |x| - r_n|<\epsilon \}$, $n\in\N$, with
        $r_n:=\pi (n + \frac14(\gamma-1))$ and sufficiently small but fixed
        $\epsilon>0$. It follows easily from \eqref{eq:diamagthm} that
        \eqref{eq:diamag} is violated for all large $n$. 
\end{proof}

\begin{remark}\label{nyanm}
The analogous statement (with the same proof) is valid for
$-1<\gamma<1$. Moreover, for exponential functions
Theorem~\ref{expthm} implies that
\begin{equation*}
  \exp(-tH_\alpha)(x,x)
=\frac1{4\pi t}\int_0^\infty
\rho_\alpha\left(\sqrt\lambda\frac{|x|}{\sqrt
    t}\right)e^{-\lambda}d\lambda
=\frac1{4\pi t} \sigma_{\alpha,\infty}\left(\frac{|x|}{\sqrt t}\right)
\end{equation*}
is \emph{strictly} less than $(4\pi t)^{-1}=\exp(-t(-\Delta))(x,x)$. 
\end{remark}

The following substitute for \eqref{eq:diamag} will be useful later on.

\begin{proposition}\label{weakdiamag}
	Let $0<\alpha<1$ and let $\phi$ be given by $\eqref{eq:momentphi}$ for some $\gamma>-1$, $\Lambda>0$. Then for all open sets $\Omega\subset\R^2$
	\begin{equation}\label{eq:weakdiamag}
		\tr\chi_\Omega\phi(H_\alpha)\chi_\Omega \leq
                R_{\gamma}(\alpha)
                \tr\chi_\Omega\phi(-\Delta)\chi_\Omega 
	\end{equation}
	with
	\begin{equation}\label{eq:r}
		R_{\gamma}(\alpha):= 
		(\gamma+1) \sup_{r\geq 0} \int_0^1 (1-\lambda)^\gamma \rho_\alpha (\sqrt\lambda r) \, d\lambda. 
	\end{equation}
\end{proposition}
	
Indeed, this is an immediate consequence of
\eqref{eq:diamagthm1}. This shows as well that the right-hand side of
\eqref{eq:r} yields the sharp constant in \eqref{eq:weakdiamag}. 

\begin{remark}\label{gammamono}
	The constant $R_{\gamma}(\alpha)$ is strictly decreasing with
        respect to $\gamma$. Indeed, following \cite{AL} we write, for
        $\gamma>\gamma'>-1$ and $0\leq\lambda\leq 1$, 
	\begin{equation*}
		B(\gamma-\gamma',\gamma'+1) (1-\lambda)^\gamma =
		\int_0^{1-\lambda}
                (1-\lambda-\mu)^{\gamma'}\mu^{\gamma-\gamma'-1}\,d\mu
	\end{equation*}
	and find, for any $r\geq 0$,
	\begin{align*}
		& B(\gamma-\gamma',\gamma'+1) \int_0^1
                (1-\lambda)^\gamma \rho_\alpha (\sqrt\lambda r) \,
                d\lambda \\ 
		=&   
		\int_0^1 \left( \int_0^{1-\mu}
                  (1-\lambda-\mu)^{\gamma'} \rho_\alpha (\sqrt\lambda
                  r)\, d\lambda \right) \mu^{\gamma-\gamma'-1} \,d\mu \\
		=&  
		\int_0^1 w(r\sqrt{1-\mu}) (1-\mu)^{\gamma'+1}
                \mu^{\gamma-\gamma'-1} \,d\mu, 
	\end{align*}
	where $w(s) := \int_0^1 (1-\lambda)^{\gamma'} \rho_\alpha
        (\sqrt\lambda s)\, d\lambda$. Since $w(s) \leq
        R_{\gamma'}(\alpha)(\gamma'+1)^{-1}$ one concludes that 
	\begin{align*}
		R_{\gamma}(\alpha)
		& \leq R_{\gamma'}(\alpha) (\gamma'+1)^{-1} (\gamma+1)
                B(\gamma-\gamma',\gamma'+1)^{-1}  
		\int_0^1 (1-\mu)^{\gamma'+1} \mu^{\gamma-\gamma'-1} \,d\mu \\
		& = R_{\gamma'}(\alpha) (\gamma'+1)^{-1} (\gamma+1)
                B(\gamma-\gamma',\gamma'+1)^{-1}  
		B(\gamma-\gamma',\gamma'+2) = R_{\gamma'}(\alpha).
	\end{align*}		
	That this inequality is actually strict, follows from the fact
        that the supremum in \eqref{eq:r} is attained for some
        $r_0\in(0,\infty)$ (see Theorem \ref{sigmathm}) and that
        $w(r_0\sqrt{1-\mu})$ is non-constant with respect to $\mu$.  
\end{remark}	

We close this section by giving numerical values for
$R_{\gamma}(\alpha)$. Note that 
\begin{align*}
  R_{\gamma}(\alpha)
	=& 2(\gamma+1)\sup_{r\ge 0}\int_0^1(1-k^2)^\gamma \rho_\alpha(k r)k\,dk \\
	=& \sup_{r\ge 0} r\int_0^1(1-k^2)^{\gamma+1}\rho_\alpha'(kr) \, dk \\
	=& \sup_{r\ge 0} r\int_0^1(1-k^2)^{\gamma+1}\left( J_{\alpha}(kr)
  J_{\alpha-1}(kr) +J_{1-\alpha}(kr)J_{-\alpha}(kr)\right)\, dk.
\end{align*}
The integral can be evaluated
by some quadrature algorithm and by Theorem~\ref{sigmathm} the
supremum is attained for finite $r$. This allows us to compute
approximate values of $R_\gamma(\alpha)$ $(=R_\gamma(1-\alpha))$, 
\begin{center}
  \begin{tabular}{l|lllll}
    & $R_\gamma(0.1)$ & $R_\gamma(0.2)$ & $R_\gamma(0.3)$ &
    $R_\gamma(0.4)$ & $R_\gamma(0.5)$ \\ 
    \hline
    $\gamma=0$ & 1.01682 & 1.03262 & 1.04422 & 1.05151 & 1.05397 \\
    $\gamma=\frac12$ & 1.01027 & 1.02050 & 1.02781 & 1.03241 & 1.03395 \\
    $\gamma=1$ & 1.00650 & 1.01351 & 1.01833 & 1.02138 & 1.02238 \\
    $\gamma=\frac32$ & 1.00417 & 1.00920 & 1.01250 & 1.01457 & 1.01524 \\
    $\gamma=2$ & 1.00267 & 1.00642 & 1.00874 & 1.01019 & 1.01065
  \end{tabular} 
\end{center}
In addition to the monotonicity in $\gamma$ (see Remark
\ref{gammamono}) the constant $R_{\gamma}(\alpha)$ seems to be strictly
increasing in $\alpha$, but we have not been able to
prove this.


\section{A magnetic Berezin-Li-Yau inequality}\label{sec:bly}

As in the introduction, let $\Omega\subset\R^2$ be a bounded domain and define
$H_\alpha^\Omega$ in $L_2(\Omega)$ through the closure of the 
quadratic form $\|(\D-\alpha \A_0)u\|^2$ on 
$C_0^\infty(\Omega\setminus\{0\})$. For eigenvalue moments of this
operator we shall prove 

\begin{theorem}\label{blythm}
	Let $0<\alpha<1$, $\gamma\geq1$ and $\Omega\subset\R^2$ be a
        bounded domain such that the operator $H_\alpha^\Omega$ has
        discrete spectrum. Then for any $\Lambda>0$, 
	\begin{align}\label{eq:blythm}
		\tr(H_\alpha^\Omega-\Lambda)_-^\gamma
		\leq R_{\gamma}(\alpha) \frac1{4\pi(\gamma+1)}
                |\Omega| \Lambda^{\gamma+1}  
	\end{align}
	with the constant $R_{\gamma}(\alpha)$ from \eqref{eq:r}.
\end{theorem}

As explained in the introduction, 
\begin{equation*}
	\frac1{4\pi(\gamma+1)} |\Omega| \Lambda^{\gamma+1} 
	= \frac1{(2\pi)^{2}}\iint_{\R^2\times\Omega} (|\xi|^2-\Lambda)^\gamma_-\,d\xi\,dx
\end{equation*}
is the semi-classical approximation for
$\tr(H_\alpha^\Omega-\Lambda)_-^\gamma$. Unfortunately, we can only
prove \eqref{eq:blythm} with an excess factor $R_{\gamma}(\alpha)$,
which is strictly larger than one by Theorem \ref{sigmathm}. Based on
numerical calculations (see next section) we conjecture that
\eqref{eq:blythm} should be valid with $R_{\gamma}(\alpha)=1$. 

\begin{proof} 
  We first note that
  \begin{equation}\label{eq:gamma0}
    \tr(H_\alpha^\Omega-\Lambda)_- 
    \le \tr\chi_\Omega(\chi_\Omega H_\alpha
    \chi_\Omega-\Lambda)_-\chi_\Omega.
  \end{equation}
  Indeed, let  $(\omega_j)$ be an orthonormal basis of eigenfunctions of
  $H_\alpha^\Omega$. Then the extension $\tilde\omega_j$ of $\omega_j$
  by zero belongs to the form domain of $\chi_\Omega
  H_\alpha\chi_\Omega$ and
  \begin{align*}
    \tr (H_\alpha^\Omega-\Lambda)_- 
    & = \sum_j ((H_\alpha^\Omega-\Lambda)\omega_j,\omega_j)_-
     = \sum_j ((\chi_\Omega H_\alpha\chi_\Omega-\Lambda)
    \tilde\omega_j,\tilde\omega_j)_- \\
    & \leq \sum_j ((\chi_\Omega H_\alpha\chi_\Omega-\Lambda)_-
    \tilde\omega_j,\tilde\omega_j) 
     = \tr\chi_\Omega(\chi_\Omega H_\alpha
    \chi_\Omega-\Lambda)_-\chi_\Omega.
  \end{align*}
  Now let $\phi$ be a convex function of the form
  $\phi(\lambda)=\int(\lambda-\mu)_-w(\mu)\,d\mu$ for some $w\geq
  0$. Then it follows from \eqref{eq:gamma0} that
  \begin{equation*}
    \tr\phi(H_\alpha^\Omega) 
    \le \tr\chi_\Omega\phi(\chi_\Omega H_\alpha \chi_\Omega)\chi_\Omega,
  \end{equation*}
  and hence by the Berezin-Lieb inequality (see \cite{B2}, \cite{L}
  and also \cite{LaS}, \cite{laptev}),
  \begin{align*}
    \tr\phi(H_\alpha^\Omega)
    \leq
    \tr\chi_\Omega\phi(\chi_\Omega H_\alpha \chi_\Omega)\chi_\Omega 
    \leq \tr\chi_\Omega\phi(H_\alpha)\chi_\Omega.
  \end{align*}
  As already noted in Remark \ref{gammamono}, the function
  $\phi(\lambda)=(\lambda-\Lambda)_-^\gamma$ is of the considered
  form. Noting that in this case
  $\tr\chi_\Omega\phi(-\Delta)\chi_\Omega =(4\pi(\gamma+1))^{-1}
  |\Omega| \Lambda^{\gamma+1}$, the assertion 
  follows from Proposition \ref{weakdiamag}.
\end{proof}


\section{Numerical experiments}\label{sec:num}

\subsection{Numerical evaluation of the magnetic Berezin-Li-Yau  constant}

In this section we present a numerical study,
somewhat in the spirit of \cite[Appendix~A]{LT1}, of the constant
$R_{\gamma}(\alpha)$ appearing in \eqref{eq:blythm}
for $H_\alpha^\Omega$ on various domains. We consider the unit disc, a
square and three annuli: 
\begin{align*}
  A&:=\{x:|x|< 1\}, 
& C&:=\{x:  1<|x|< 1.1\}, \\ 
  B&:=\{x: \max(|x_1|, |x_2|)< 1\}, 
& D&:=\{x: 1<|x|< 2\},    \\ 
&&E&:=\{x: 1<|x|< 11\}.
\end{align*}
To expose possible diamagnetic effects we perform identical
experiments for $\alpha=0$ and $\alpha=0.2$ throughout. If the
eigenvalues $\lambda_1,\lambda_2,\ldots$ are known, the task is to determine 
\begin{equation}\label{ltdef}
  R_\gamma=\sup_\Lambda r_\gamma(\Lambda),
\qquad\text{where}\ r_\gamma(\Lambda):=
\frac{4\pi(\gamma+1)}{|\Omega|\Lambda^{\gamma+1}}
\sum_{j:\ \lambda_j<\Lambda} (\Lambda-\lambda_j)^\gamma.
\end{equation}
Since  the sum is simply the counting function if $\gamma=0$, $r_0$ is
decreasing for all $\Lambda$ but the eigenvalues. 
In the physically  important
case $\gamma=1$, the quotient $r_1$ is continuous but $r_1'$ has jump
increases at the eigenvalues. Figures~\ref{fig:ltbot} and \ref{fig:ltinf}
show, respectively, the schematic behaviour of $r_\gamma$ near the bottom of the
spectrum and for large $\Lambda$ in the case where $R_\gamma=1$
(in fact, the spectrum of $H_0^B$ has been considered).
Figure~\ref{fig:ltnoncl} depicts $r_0(\Lambda)$ for an operator 
with $R_\gamma>1$ (see \cite[Thm.~3.2]{amh}),
namely the Schr\"odinger operator $H_\alpha-|x|^{-1}$. 
The  $\times$ symbols on the $x$ axis indicate the eigenvalue loci.

\subsection{Computation of the eigenvalues}

By separation of variables the spectrum of $H_\alpha^A$ is the union
of the spectra of the family of one-dimensional problems 
\begin{equation*}
  -u''+\frac{(n-\alpha)^2-1/4}{r^2}u=\lambda u, \qquad u(1)=0, \qquad u\in L^2(0,1),
\end{equation*}
parametrized by $n\in\Z$. $\sqrt{r}J_{|n-\alpha|}(\sqrt{\lambda}r)$ is an
eigenfunction provided $\lambda$ is chosen to that
$J_{|n-\alpha|}(\sqrt{\lambda})=0$, i.e., the eigenvalues are squares of
the zeros of the Bessel function of the first kind. Using standard
numerical procedures for this task we compute all eigenvalues
below 10,000 (approximately 2,500) of $H_0^A$ and all eigenvalues below
50,000  (approx. 12,000) of 
$H_{0.2}^A$. The lowest eigenvalues of the two operators can be viewed
in Figure~\ref{fig:Aeigv}. The reader interested in the dependence of
the eigenvalues on the parameter $\alpha$ may find 
Figure~\ref{fig:rf6} useful, where the lowest eigenvalues of
$H^A_\alpha$ as functions of $\alpha\in[0,\frac12]$ are shown. We recall
that $H^A_\alpha$ and $H^A_{1-\alpha}$ share the same spectrum.
 
Similarly, the eigenvalue problems for $H_\alpha^C$, $H_\alpha^D$ and
$H_\alpha^E$ are reduced to
\begin{equation*}
  -u''+\frac{(n-\alpha)^2-1/4}{r^2}u=\lambda u, \qquad u(r_1)=0=u(r_2),
  \qquad u\in L^2(r_1,r_2).
\end{equation*}
In order that the general solution
$\sqrt{r}(c_1J_{|n-\alpha|}(\sqrt{\lambda}r)+c_2Y_{|n-\alpha|}(\sqrt{\lambda}r))$ 
satisfy the Dirichlet conditions we need to fix the ratio $c_1/c_2$ and
take $\lambda$ such that 
\begin{equation*}
  J_{|n-\alpha|}(\sqrt{\lambda}r_1)Y_{|n-\alpha|}(\sqrt{\lambda}r_2)
=J_{|n-\alpha|}(\sqrt{\lambda}r_2)Y_{|n-\alpha|}(\sqrt{\lambda}r_1).
\end{equation*}
Equivalently, $\lambda=x^2/r_1^2$ is an eigenvalue whenever $x$ is a zero of
the cross-product $J_\nu(t)Y_\nu(\mu t)-J_\nu(\mu
t)Y_\nu(t)$ for $\mu=r_2/r_1$, $\nu=|n-\alpha|$. Since the
amplitude of the oscillation grows approximately as $\mu^\nu/\nu$,
the computation of the zeros is nontrivial for large orders. By
searching in the negative $t$ direction and making use of the simple
fact that the asymptotic spacing of the zeros is $\pi/(\mu-1)$
\cite[9.5.28]{as}, 
we circumvent much of these difficulties in comparison to routines
such as \textsc{Mathematica}'s  \emph{BesselJYJYZeros}. This allows us
to compute, for $\alpha=0$ and $0.2$, all eigenvalues below 10,000
(approx. 400)
of $H_\alpha^C$, all eigenvalues below 1,000 (approx. 550) of $H_\alpha^D$ and
all eigenvalues below 300 (approx. 9,000) of $H_\alpha^E$.

$H_0^B$ is simply the Dirichlet Laplacian on $[-1,1]^2$, so for all $k,l\in\N$,
$$\sin\left({k\pi}\frac{x+1}2\right)\sin\left({l\pi}\frac{y+1}2\right)$$
is an eigenfunction and $\frac{\pi^2}4(k^2+l^2)$  an eigenvalue. For
$\alpha\neq 0$ the absence of radial symmetry  prevents us
from solving the eigenvalue problem for $H_\alpha^B$ exactly. We
therefore resort to the finite element method as implemented in
\textsc{Comsol Multiphysics} to compute approximate eigenvalues. A
8,800-element mesh, gradually refined near the singularity, is used to
compute the 250 lowest eigenvalues, all less than 840. Providing an a
priori bound on the computational error would be beyond the scope of
this article; it is an encouraging fact, however, that the approximate
eigenvalues of $H^A_{0.2}$ do not differ from the exact ones by more
than 0.8~\%. Figure~\ref{fig:rf5} is the analogue of
Figure~\ref{fig:Aeigv} for $H^B_\alpha$.

\subsection{Outcome of the experiments}
We are obviously constrained to computing \eqref{ltdef} with $\Lambda$
restricted to a finite interval but provided this interval is
sufficiently large it should be possible to tell the behaviors depicted in
Figures~\ref{fig:ltinf} and \ref{fig:ltnoncl} apart. Particularly if
the function $r_\gamma$ tends steadily to $1$ from below we can conclude 
with certainty that its supremum will eventually be this number in
virtue of the eigenvalue asymptotics. It turns out that \emph{all the
  operators we consider, whether $\alpha=0$ or $0.2$, seem to obey a
  semi-classical eigenvalue estimate with unit constant for both
  $\gamma=0$ and $1$.} 

What however differs is the rate of convergence to 1, e.g., $r_\gamma$
increases slower for $H_\alpha^C$ than for $H_\alpha^E$ since
$|C|<|E|$, implying that the
point spectrum of the former operator is sparser. To see the influence
of $\alpha$ on $r_\gamma$ we invite the reader to compare
Figures~\ref{fig:rf3b}, \ref{fig:rf2a} and \ref{fig:rf2b}. Note that
$r_\gamma$ has the most irregular behavior in the cases $\alpha=0$ and
$\alpha=0.5$ due to the high degree of degeneracy of the eigenvalues,
cf.~Figure~\ref{fig:rf6}. Moreover we
have the impression, from looking at Figures~\ref{fig:ltbot},
\ref{fig:rf3b} and \ref{fig:rf3a}, that the shape of the domain has
a  limited influence on $r_\gamma$. All three
plots suggest that $R_0=1$ but this  has not been proved so far even
for $H^A_0$, the Laplacian on the unit disc!


\bibliographystyle{amsalpha}

\begin{thebibliography}{BaEvLe}

\bibitem[AhBo]{ab}Y.~Aharonov, D.~Bohm, \emph{Significance of Electromagnetic
  Potentials in the Quantum Theory.} Phys. Rev. II \textbf{115}
  (1959), 485--491
\bibitem[AL]{AL} M. Aizenman, E. Lieb, \textit{On semiclassical
  bounds for eigenvalues of Schr\"odinger operators}. Phys. Lett. A
  \textbf{66} (1978), no. 6, 427--429.
\bibitem[AbSt]{as}M.~Abramowitz, I.~A.~Stegun, \emph{Handbook of
    mathematical functions withs formulas, graphs, and mathematical
    tables}, National Bureau of Standards Applied Mathematics Series
  55, Washington D.C., 1964.
\bibitem[BaEvLe]{bel}A.~A.~Balinsky, W.~D.~Evans, R.~T.~Lewis, \emph{On the
  number of negative eigenvalues of Schr\" odinger operators with an
  Aharonov-Bohm magnetic field}. Proc. R. Soc. Lond. A
  \textbf{457} (2001) 2481--2489
\bibitem[B1]{B} F.~A.~Berezin, \textit{Covariant and contravariant
    symbols of operators} [Russian]. Math. USSR Izv. {\bf 6} (1972),
  1117--1151. 
\bibitem[B2]{B2} F.~A.~Berezin, \textit{Convex functions of operators}
  [Russian]. Mat. Sb. \textbf{88} (1972), 268--276. 
\bibitem[EkF]{ef}T.~Ekholm, R.~L.~Frank, \emph{On Lieb-Thirring
    inequalities for Schr\"odinger operators with virtual level.} 
Comm. Math. Phys. \textbf{264} (2006), no.~3, 725--740. 
\bibitem[ELo]{ELo} L.~Erd\H os, M.~Loss, private communication. 
\bibitem[ELoV]{ELV} L.~Erd\H os, M.~Loss, V.~Vougalter,
  \textit{Diamagnetic behavior of sums of Dirichlet eigenvalues}. 
    Ann. Inst. Fourier {\bf 50} (2000), 891--907.
\bibitem[FLoW]{FLW} R.~L.~Frank, M.~Loss, T.~Weidl, in preparation.
\bibitem[H]{amh}A.~M.~Hansson, \emph{On the spectrum and eigenfunctions of the
  Schr\"odinger operator with Aharonov-Bohm magnetic field.}
  Int. J. Math. Math. Sci. \textbf{23} (2005), 3751--3766.
\bibitem[HuS]{hus}D.~ Hundertmark, B.~Simon, 
\emph{A diamagnetic inequality for semigroup differences.} 
J. Reine Angew. Math. \textbf{571} (2004), 107--130. 
\bibitem[La]{laptev}A.~Laptev, \emph{Dirichlet and Neumann eigenvalue
  problems on domains in Euclidean spaces.}  J. Funct. Anal.  \textbf{151}
  (1997),  no.~2, 531--545. 
\bibitem[LaSa]{LaS} A.~Laptev, Yu.~Safarov, \textit{A generalization of
    the Berezin-Lieb inequality}. Amer. Math. Soc. Transl. (2)
  \textbf{175} (1996), 69--79. 
\bibitem[LaW1]{LW} A. Laptev, T. Weidl, \textit{Sharp Lieb-Thirring
    inequalities in high dimensions}. Acta Math. \textbf{184} (2000),
  no. 1, 87--111. 
\bibitem[LaW2]{LW2} A. Laptev, T. Weidl, \textit{Recent results
  on Lieb-Thirring inequalities}. Journ\'ees "\'Equations aux D\'eriv\'ees
  Partielles" (La Chapelle sur Erdre, 2000), Exp. No. XX,
  Univ. Nantes, Nantes, 2000. 
\bibitem[LiY]{LY} P.~Li, S.-T.~Yau, \textit{On the Schr\"odinger
    equation and the eigenvalue problem}. Comm. Math. Phys. {\bf 88}
  (1983), 309--318. 
\bibitem[L]{L} E.~H.~Lieb, \textit{The classical limit of quantum spin
    systems}. Comm. Math. Phys. \textbf{31} (1973), 327--340. 
\bibitem[LT]{LT1} E.~H.~Lieb, W.~Thirring, \textit{Inequalities for the
    moments of the eigenvalues of the Schr\"odinger Hamiltonian and
    their relation to Sobolev inequalities}. Studies in Mathematical
  Physics, 269--303. Princeton University Press, Princeton, NJ, 1976.
\bibitem[Lu]{luke}Y.~L.~Luke, \emph{Integrals of Bessel Functions},
  McGraw-Hill, New York-London-Toronto, 1962.
\bibitem[MOR]{mor}M.~Melgaard, E.-M.~Ouhabaz, G.~Rozenblum,
  \emph{Negative discrete spectrum 
of perturbed multivortex Aharonov-Bohm Hamiltonians.}  Ann. Henri
Poincar\'e  5  (2004),  no.~5, 979--1012. Erratum: \emph{ibid.} \textbf{6} 
(2005), no.~2, 397--398.
\bibitem[P]{P} G.~P\'olya, \textit{On the eigenvalues of vibrating
    membranes}. Proc. London Math. Soc. {\bf 11} (1961), 419--433. 
\bibitem[ReS]{ReSi} M. Reed, B. Simon, \textit{Methods of Modern
    Mathematical Physics}, vol. 4. Academic Press, New York-London, 1978. 
\bibitem[Ru]{ruij}S.~N.~M.~Ruijsenaars, 
\emph{The Aharonov-Bohm effect and scattering theory},
Ann. Physics \textbf{146} (1983), no.~1, 1--34.
\bibitem[S]{simo} B.~Simon, \emph{Trace ideals and their applications}, second
ed. Mathematical Surveys and Monographs \textbf{120}, American
Mathematical Society, Providence, RI, 2005. 
\bibitem[Ti]{T}
E.~C.~Titchmarsh, \emph{Theory of Fourier Integrals}, second ed.
Oxford University Press, Oxford, 1948.

\end{thebibliography}

\begin{figure}[p]
  \includegraphics[width=\textwidth, height=8.1cm]{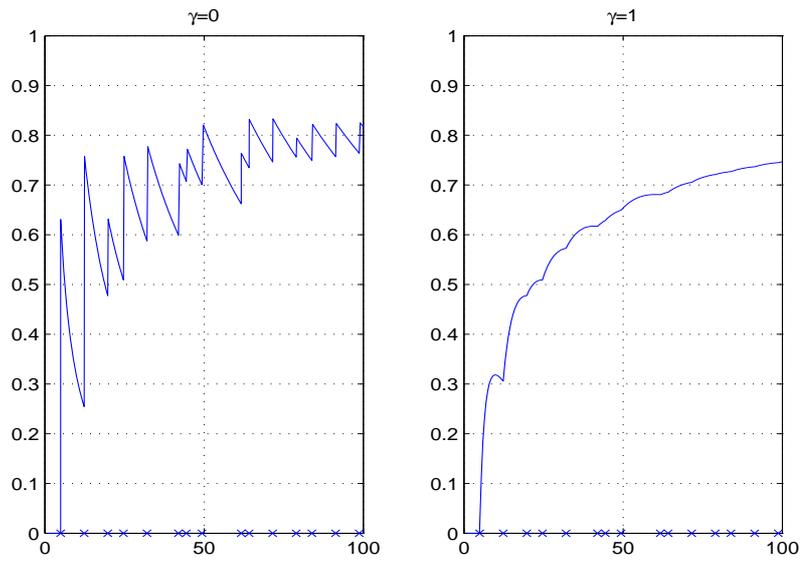}\centering
  \caption{$r_\gamma$ for $H_0^B$ near the bottom of the
    spectrum}
  \label{fig:ltbot}
\end{figure}
\begin{figure}[p]
  \includegraphics[width=\textwidth, height=8.1cm]{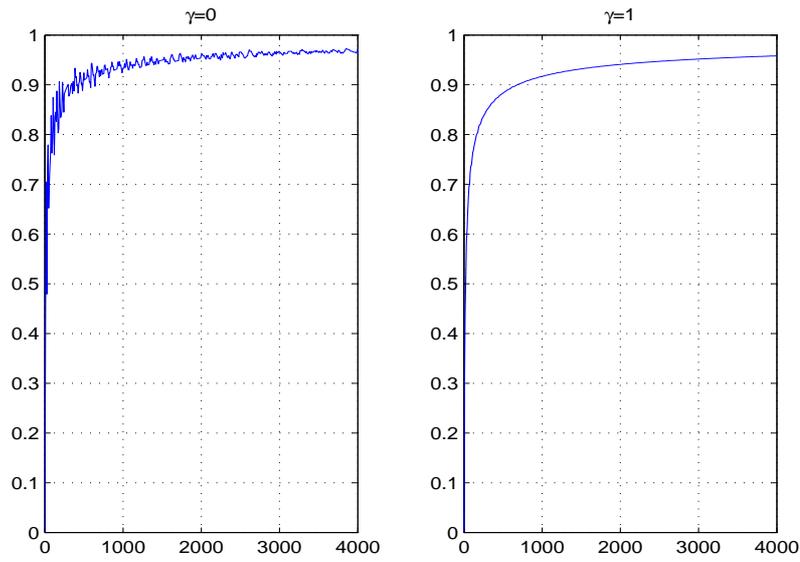}\centering
  \caption{ $r_\gamma(\Lambda)$ for $H_0^B$ for large $\Lambda$}
  \label{fig:ltinf}
\end{figure}
\begin{figure}[p]
  \includegraphics[width=\textwidth, height=8.1cm]{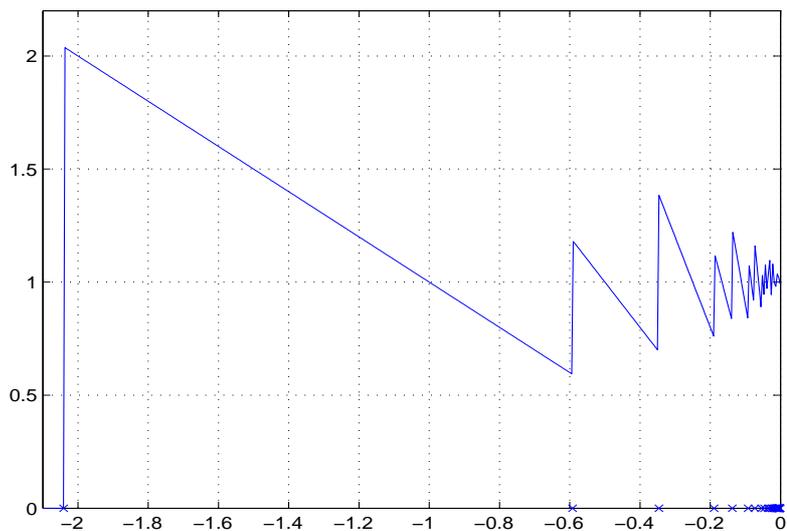}\centering
  \caption{ $r_0$ for an operator with $R_0>1$}
  \label{fig:ltnoncl}
\end{figure}
\begin{figure}[p]
  \includegraphics[width=\textwidth, height=8.1cm]{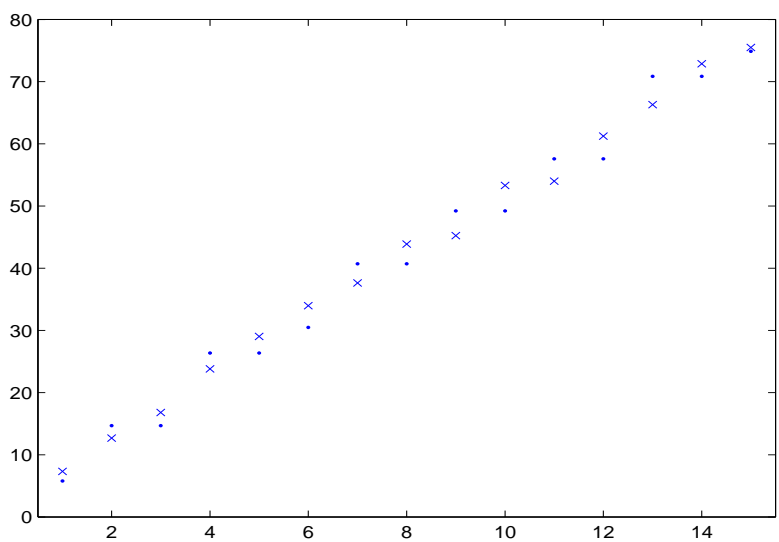}\centering
  \caption{Lowest eigenvalues of $H_0^A$ ($\cdot$) and of $H_{0.2}^A$ ($\times$)}
  \label{fig:Aeigv}
\end{figure}
\begin{figure}[p]
  \includegraphics[width=\textwidth, height=8.1cm]{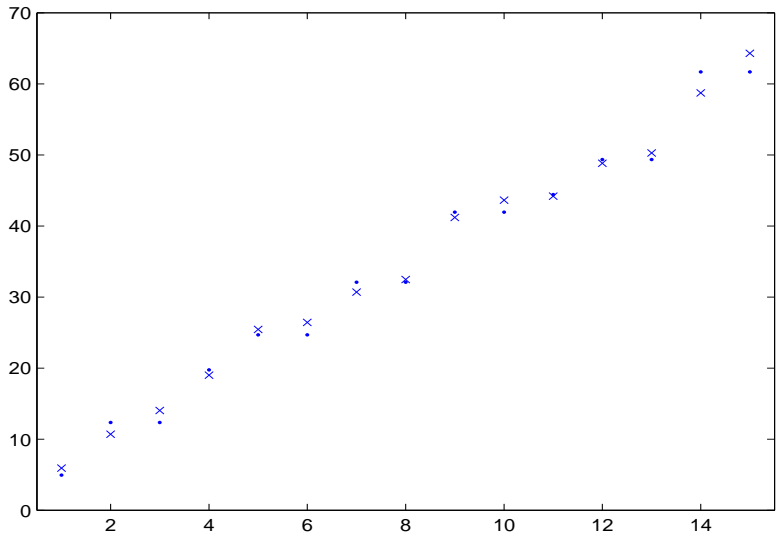}\centering
  \caption{Lowest eigenvalues of $H_0^B$ ($\cdot$) and of $H_{0.2}^B$ ($\times$)}
  \label{fig:rf5}
\end{figure}\begin{figure}[p]
  \includegraphics[width=\textwidth, height=8.1cm]{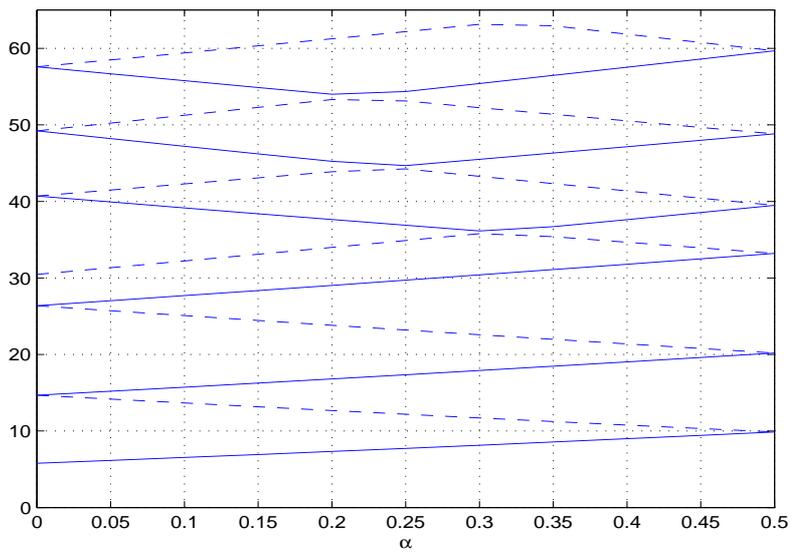}\centering
  \caption{Lowest eigenvalues of $H_\alpha^A$ as function of $\alpha$}
  \label{fig:rf6}
\end{figure}
\begin{figure}[p]
  \includegraphics[width=\textwidth, height=8.1cm]{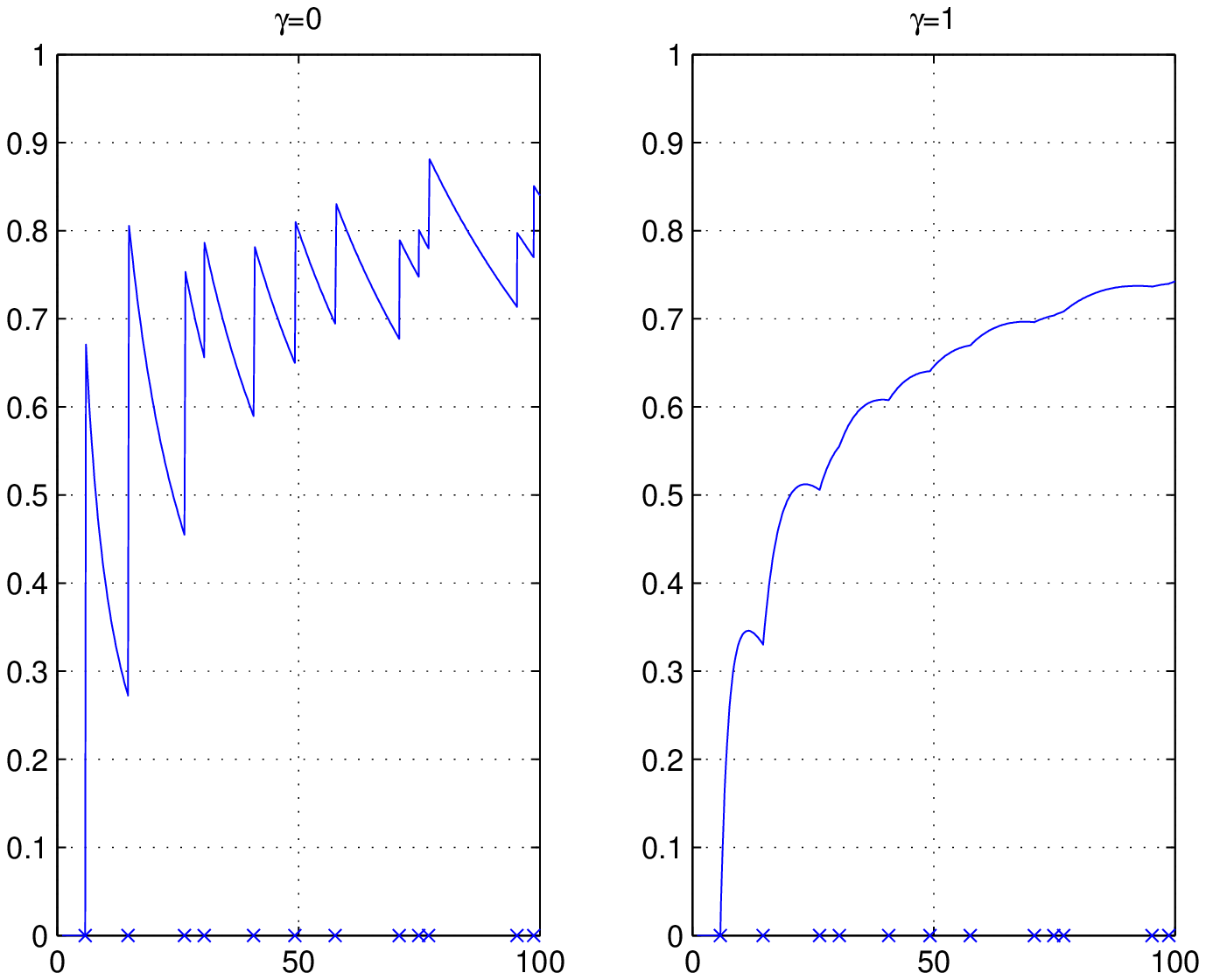}\centering
  \caption{$r_\gamma$ for $H_0^A$ near the bottom of the
    spectrum}
  \label{fig:rf3b}
\end{figure}
\begin{figure}[p]
  \includegraphics[width=\textwidth, height=8.1cm]{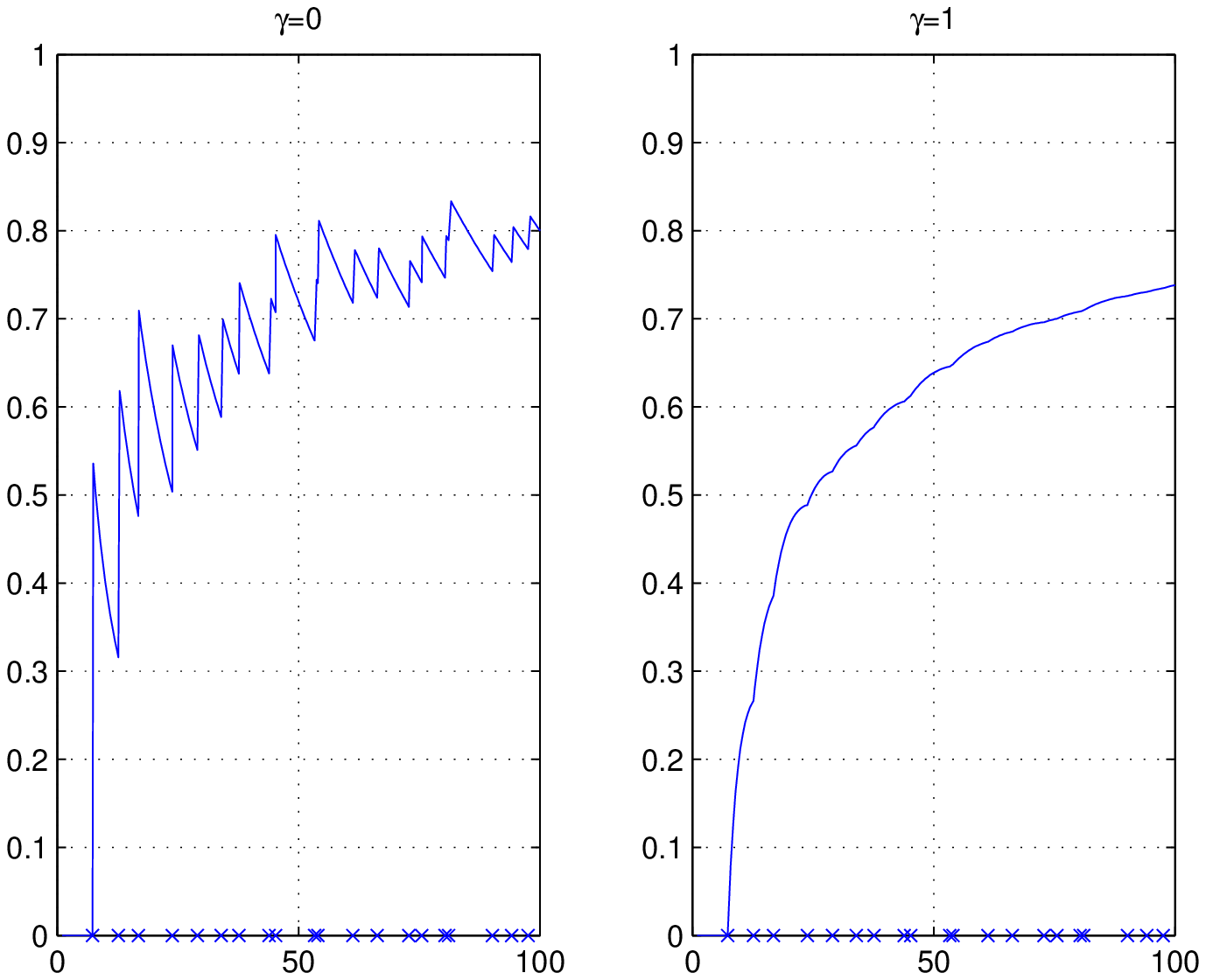}\centering
  \caption{$r_\gamma$ for $H_{0.2}^A$ near the bottom of the
    spectrum}
  \label{fig:rf2a}
\end{figure}\begin{figure}[p]
  \includegraphics[width=\textwidth, height=8.1cm]{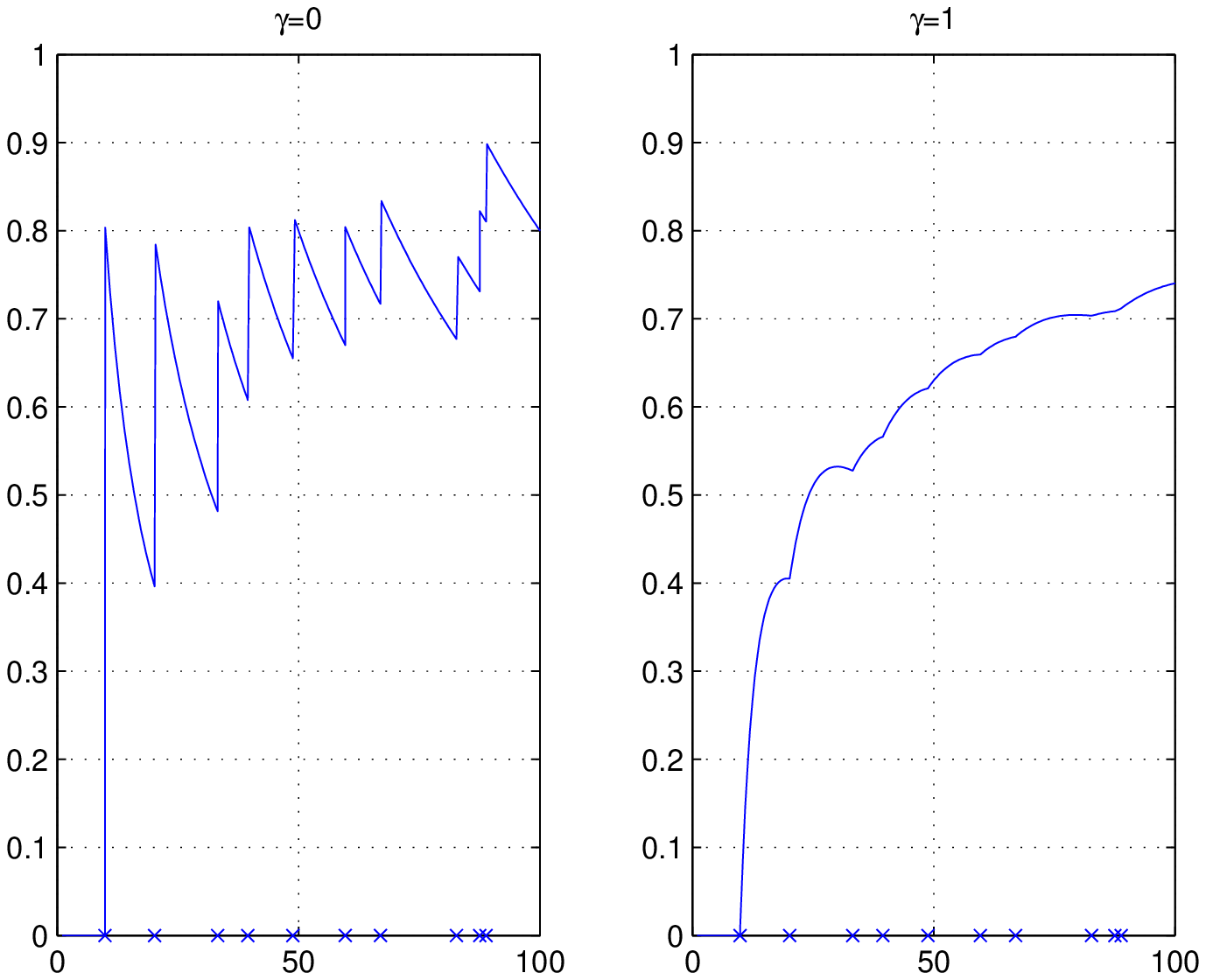}\centering
  \caption{$r_\gamma$ for $H_{0.5}^A$ near the bottom of the
    spectrum}
  \label{fig:rf2b}
\end{figure}
\begin{figure}[p]
  \includegraphics[width=\textwidth, height=8.1cm]{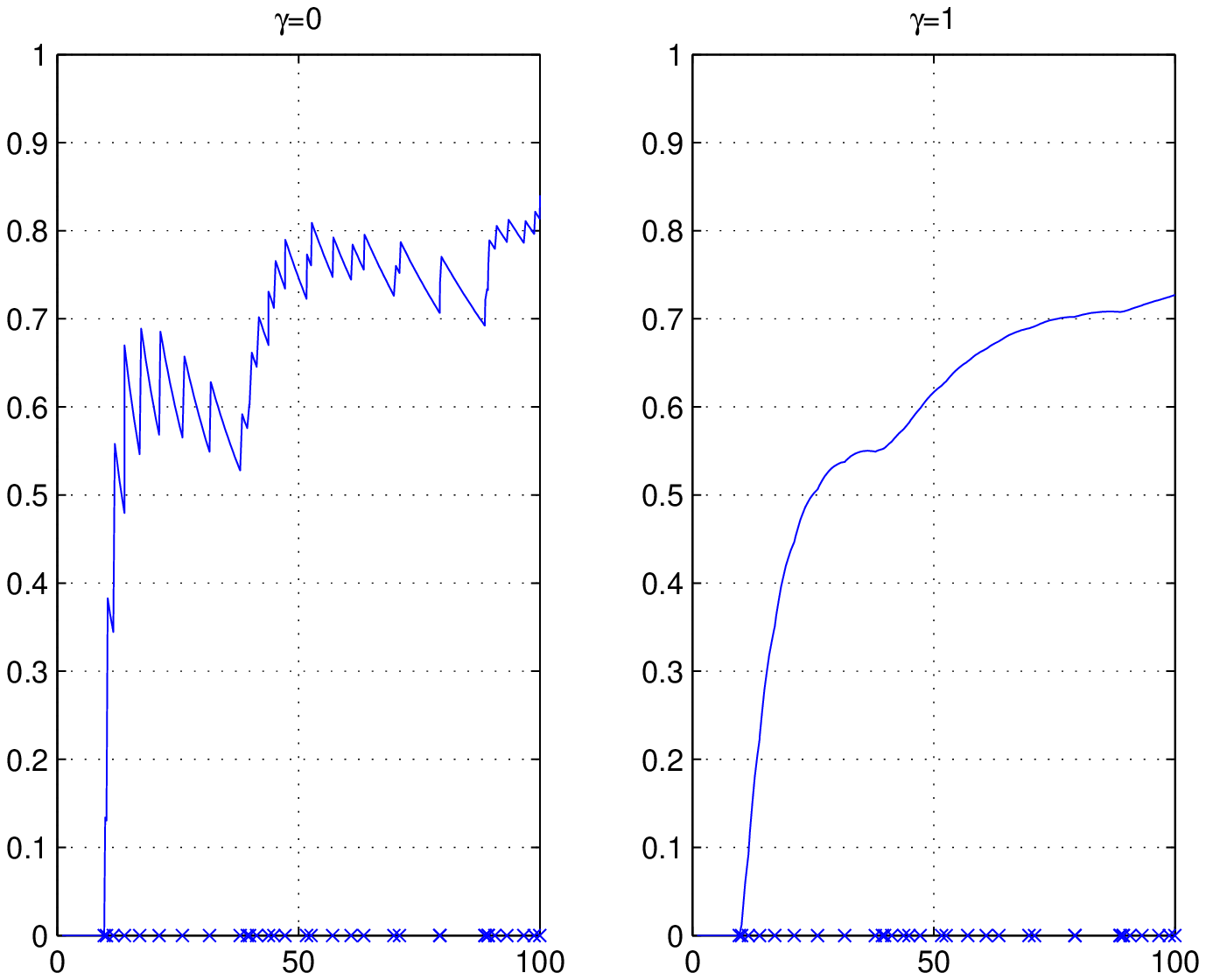}\centering
  \caption{$r_\gamma$ for $H_0^D$ near the bottom of the
    spectrum}
  \label{fig:rf3a}
\end{figure}

\end{document}